\documentclass[conference]{IEEEtran}

\usepackage{graphicx}
\usepackage{wrapfig}
\usepackage{epsfig}
\usepackage{graphicx}
\usepackage{color}
\usepackage{mdframed}
\usepackage{mathrsfs} 
\usepackage{amsfonts}
\usepackage{latexsym}
\usepackage{amsmath}
\usepackage{amssymb}
\usepackage{color}
\usepackage{float} 
\usepackage{floatflt} 
\usepackage{caption}
\usepackage{subcaption}

\usepackage{pdfpages}

\newcommand{\suppress}[1]{}

\newtheorem{theorem}{Theorem}
\newtheorem{lemma}{Lemma}

\newtheorem{definition}{Definition}
\newtheorem{corollary}{Corollary}
\newtheorem{statement}{Statement}

\def\e{\varepsilon}

\def\cX{\mbox{$\cal{X}$}}

\def\cF{\mbox{$\cal{F}$}}
\def\cG{\mbox{$\cal{G}$}}

\def\cR{\mbox{$\cal{R}$}}

\def\e{\varepsilon}

\newcommand{\cI}{{{\cal I}}}

\newcommand{\bW}{{{\bf W}}}

\newcommand{\bR}{{{\bf R}}}

\def\01{\{0,1\}}

\newcommand{\remove}[1]{}

\newcommand{\Gl}{  {G^{\lambda,e}}}
\newcommand{\Il}{  {\mathcal{I}^{\lambda,e}}}
\newcommand{\fe}{  {f_{\tt \cI}(\lambda)}}

\begin{document}
\IEEEoverridecommandlockouts
\title{Edge removal in undirected networks}

\author{Michael Langberg\ \ \ \ \ \ \ \ \  \ \ \ \ \ \ \ \ \ 
Michelle Effros
\thanks{M. Langberg is with the Department of Electrical Engineering at The State University of New York at Buffalo.  
Email : \texttt{mikel@buffalo.edu}}
\thanks{M. Effros is with the Department of Electrical Engineering at the California Institute of Technology.
Email : \texttt{effros@caltech.edu}}
\thanks{This work is supported in part by NSF grants CCF-1817241 and CCF-1909451.}
}

\maketitle

\begin{abstract}
The edge-removal problem asks whether the removal of a $\lambda$-capacity edge from a given network can decrease the communication rate between source-terminal pairs  by more than $\lambda$. 
In this short manuscript, we prove that for undirected networks, removing a $\lambda$ capacity edge decreases the rate by $O(\lambda)$.
Through previously known reductive arguments, here newly applied to undirected networks, our result implies that the zero-error capacity region of an undirected network equals its vanishing-error capacity region.
Whether it is possible to prove similar results for directed networks remains an open question.
\end{abstract}

\section{Introduction} 
The {\em edge removal problem}, defined and studied in \cite{HEJ:10,JEH:11}, 
aims to quantify the loss in capacity 
that results from the removal of a single edge 
(i.e., a point-to-point channel) 
from a given network coding instance. 
For some network coding instances, it is known that the removal of an edge of capacity $\lambda$ can decrease the rate of communication for each source-receiver pair by at most $\lambda$  \cite{HEJ:10,JEH:11}. These instances include networks with collocated sources, networks in which we are restricted to perform linear encoding, networks in which the edges removed are connected to terminals with no out going edges, as well as other families of network coding instances. 
However, whether the removal of an edge of capacity $\lambda$ decreases the rate of communication for each source-receiver pair by at most $\lambda$ for {\em any} network coding instance remains an intriguing open problem connected to a spectrum of (at times seemingly unrelated) questions in the context of network communication (see, e.g., \cite{CG10,chan2014network, LE11, LE:12, WLE13, WLE:15,WLE:16,LE:19}). 

In this work we study the edge removal problem on {\em undirected} networks. In an undirected network, the information on any edge $e=(u,u')$, can travel from $u$ to $u'$ and/or from $u'$ to $u$, as long as the sum of the rates in both directions do not exceed the edge capacity. 
Undirected networks have seen several studies in the context of network coding (e.g., \cite{li2004networka,li2004networkb,harvey2004comparing,li2006achieving,jain2006capacity,chekuri2006average,harvey2006capacity,li2009constant,maheshwar2012bounding,langberg2009multiple,cai2015network, haeupler2019network, braverman2016network}).
To date, the arguably most well known open question regarding network coding in undirected networks concerns the maximal potential benefit in rate that one can obtain in multiple-unicast instances when comparing communication via network coding to communication without coding, i.e., the {\em coding advantage} in undirected multiple unicast networks. 
It is conjectured in \cite{li2004networka,li2004networkb,harvey2004comparing} that no such advantage exists. This conjecture has been confirmed on several special cases (e.g.,\cite{li2004networkb,jain2006capacity,harvey2006capacity}) but remains an open question in full generality.

Our work is structured as follows.
In Section~\ref{sec:model}, we present our model and define a number of statements regarding the edge-removal problem.
Our main results and analysis are given in Section~\ref{sec:main}.
In Theorem~\ref{the:main}, we show that for any undirected network coding instance $\cI$ there exists a constant $c$ such that the removal of an edge of capacity $\lambda$ from $\cI$ reduces the rate between source-terminal pairs by at most $c\lambda$.
We then derive two immediate corollaries to Theorem~\ref{the:main}. In Corollary~\ref{cor:main}, we prove the so-called asymptotic edge removal statement on undirected instances. This statement asserts that removing an edge of negligible capacity has a negligible effect on the rate between source-terminal pairs. Building on prior work \cite{LE11}, in Corollary~\ref{cor:zero} we prove that the zero-error capacity region and the vanishing-error capacity region of undirected network coding instances are equal.

\section{Model} 
\label{sec:model}

Throughout the paper, the size of a finite set $S$ is denoted by $|S|$. 
For any positive real $k$, $[k]$ denotes the set $\{1,...,\lfloor k \rfloor\}$.
We use bold letters to denote vectors; for example, $\bR=(R_1,...,R_{k})$ is a vector of dimension $k$ and $R_i$ is the $i^{th}$ element of vector $\bR$.
We define $\mathbf{R}-\gamma$ as $((R_1-\gamma)^+,...,(R_{k}-\gamma)^+)$ where
$(R-\gamma)^+ = \max\{0,R-\gamma\}$.
For $\alpha>0$ and a set $\cR$ of real vectors, the set $\alpha \cR$ refers to the set obtained by multiplying each vector in $\cR$ by $\alpha$.

\subsection{Network Coding Instances and Network Codes}
An undirected instance ${\mathcal I}=(G,S,D,M)$ of the network coding problem includes an undirected network $G=(V,E)$, a 
vector of $k$ source nodes $S =(s_1,\dots, s_k) \in V^k$, a vector of terminal nodes $D = (d_1, \dots, d_r) \in  V^r$, and a binary requirement matrix $M=[m_{ij}]$ in which $m_{ij}=1$ if and only if the message of source $s_i$ is requested by terminal $d_j$. 
Source node $s_i \in S$ holds message random variable $W_i$ demanded by terminals $\{d_j \mid m_{ij}=1\}$. Each edge $e \in E$ has an associated capacity $\lambda_e$.

We here assume that communication occurs in $N$ {\em rounds}, and in each round, every edge $e \in E$ carries a message over an alphabet $\cX^n_e$ of size $\lfloor 2^{\lambda_e n}\rfloor$. 
We call $N$ the {\em outer blocklegth} and $n$ the {\em inner blocklegth}.
Namely, we think of communication over edge $e$ in terms of symbols over the alphabet $\cX^n_e$ corresponding to $n$ channel uses.
We thus use the term {\em time-step} to refer to each round of communication.

More formally, for an outer blocklength $N$, and an inner blocklegth $n$, network code 
$$({\mathcal F},\mathcal{G})=(\{\overrightarrow{f_{e,t}}\},\{\overleftarrow{f_{e,t}}\},\{g_j\})$$ 
is an assignment of encoding functions 
$\{\overrightarrow{f_{e,t}}\}$ and $\{\overleftarrow{f_{e,t}}\}$ for every time step $t \in [N]$ and each edge $e\in E$ and a decoding function $g_j$ to each terminal $d_j\in D$. 
At each time step $t$ and for each edge $e=(u,u')$ the alphabet $\cX^n_e$ is represented by two sets 
$$\overrightarrow{\cX^n_{e,t}} \ \mbox{and} \ \overleftarrow{\cX^n_{e,t}} \ \mbox{such that}\ |\overleftarrow{\cX^n_{e,t}}|\cdot |\overrightarrow{\cX^n_{e,t}}| \leq |\cX^n_e|.$$
At each time step $t$ and for each edge $e=(u,u')$
the edge message
$\overrightarrow{X^n_{e,t}} \in \overrightarrow{\cX^n_{e,t}}$ from $u$ to $u'$ and the message $\overleftarrow{X^n_{e,t}} \in \overleftarrow{\cX^n_{e,t}}$ from $u'$ to $u$ are equal to the evaluation of encoding functions $\{\overrightarrow{f_{e,t}}\}$ and $\{\overleftarrow{f_{e,t}}\}$ on inputs $X^n_{{\rm In}(u),[t-1]}$ and $X^n_{{\rm In}(u'),[t-1]}$, respectively. 
Here, for a generic node $u_0$, and time $t$, 
$$X^n_{{\rm In}(u_0),[t]}=(\overrightarrow{X^n_{e',t'}}:e' = (v,u_0) \in E, t' \leq t), (W_i: u_0=s_i)$$  
captures all information available to node $u$ at time $t$.
The evaluation of decoding function $g_{j}$ on the vector of random variables $X_{{\rm In}(d_j),[N]}$ equals the reproduction of message random variables $(W_i: m_{i,j}=1)$ requested at terminal node $d_j \in D$. 

Suppose that we are given rate vector $\bR=(R_1,\dots,R_k)$, constant $\e \in [0,1]$, and positive integers $n,N$. Instance $\cI$ of the network coding problem is said to be $(\e,\mathbf{R},n,N)$-feasible if for $W_i$ uniformly distributed over $[2^{R_iNn}]$ (for $i \in [k]$) there exists a network code $({\mathcal F,G})$ with inner-blocklength $n$ and outer-blocklength $N$ such that, 
with probability at least $1-\e$, for each $d_j \in D$ the output of decoding function $g_j$ equals $(W_i: m_{i,j}=1)$. 

\begin{definition}
[Capacity region]
\label{def:cap}
The capacity region of $\cI$, denoted by $\mathcal{R}(\cI)$, is the set of all rate vectors $\mathbf{R}$ such that for all $\e>0$ and all ${\Delta}>0$ there exist infinitely many blocklengths $n$ and infinitely many blocklengths $N$ such that $\cI$ is $(\e, \mathbf{R}-{\Delta},n,N)$-feasible. 
\end{definition}

\begin{definition}
[Zero-error capacity region]
\label{def:zcap}
The zero-error capacity region of $\cI$, denoted by $\mathcal{R}_0(\cI)$, is the set of all rate vectors $\mathbf{R}$ such that for all ${\Delta}>0$ there exist infinitely many blocklengths $n$ and infinitely many blocklengths $N$ such that $\cI$ is $(0, \mathbf{R}-{\Delta},n,N)$-feasible. 
\end{definition}
\vspace{2mm}

\noindent
{\bf Some remarks are in place.} 
For directed acyclic networks, our Definitions~\ref{def:cap} and \ref{def:zcap}, which use both inner and outer blocklengths, are equivalent to the standard definitions of capacity, e.g., \cite{ACLY00}, in which for a single blocklength parameter $\tilde{n}$, each edge $e$ of capacity $\lambda_e$ can communicate a message in $\lfloor 2^{\lambda_e \tilde{n}} \rfloor$. In this equivalence, the blocklegth $\tilde{n}$ equals the product $Nn$.

Our notion of inner and outer blocklengths stems from two aspects of cyclic networks. Primarily, given the cyclic dependence of information flowing through the network, communication is often defined in rounds, in which each round of communication depends on the information obtained through previous rounds. Hence we employ the outer blocklegth $N$.
Secondly, to accommodate networks with edge capacities $\lambda_e$ for which $\lfloor 2^{\lambda_e} \rfloor=0$ (e.g., the bounding model for a binary symmetric channel from \cite{koetter2011theory}), we consider communication over {\em sub-rounds}  in which outgoing edge messages are aggregated over an inner blocklegth of size $n$. The rate $R$ is normalized by the product $Nn$.

Operationally speaking, our notion of inner and outer blocklengths governs the cyclic dependence of coding operations over time, where for the inner-blocklength $n$ the cyclic dependence is temporarily broken. Therefore, our definitions imply  tradeoffs between the outer-blocklegth $N$ and inner-blocklegth $n$. For example, if all edge capacities are integral, then any network code that is $(\e, \mathbf{R}-{\Delta},n,N)$-feasible is also $(\e, \mathbf{R}-{\Delta},1,nN)$-feasible, but the other direction does not necessarily hold. 

We now address two lemmas that are useful in our analysis.

\begin{lemma}[\cite{LE11}]
\label{lemma:eps}
Let $\cI=(G,S,D,M)$. Let $\bR \in \cR(\cI)$. Then for any $\Delta>0$ there exist infinitely many blocklengths $n$ and infinitely many blocklengths $N$ such that $\cI$ is  $(\e, \bR-\Delta,n,N)$-feasible with $\e \leq 1/\max^2(n,N)$.
\end{lemma}

\begin{proof}
Our proof, presented here for completeness, follows the line of proof given in Claim 2.1 of \cite{LE11}.
Let $\bR \in \cR(\cI)$. Let $\e>0$ and $\Delta>0$. Consider an $(\e,\bR-\Delta,n,N)$-feasible code $(\cF,\cG)$ for $\cI$ with $n$ and $N$ sufficiently large. We first show below, using $m$ parallel executions of $(\cF,\cG)$ with a carefully chosen {\em outer code}, for any $\Delta'>0$ and any sufficiently large $m$, that $\cI$ is  $(\e',\bR-\Delta',nm,N)$-feasible with $\e' \leq 1/(nm)^2$. 

We start by setting some notation.
Let $\tilde{W}=(\tilde{W}_1,\dots,\tilde{W}_k)$ be the messages corresponding to code $(\cF,\cG)$ with  $\tilde{W}_{i} \in [2^{(R_i-\Delta)n}]$ for $i=1,\dots,k$.
Let $A$ be the subset of source messages $\tilde{W}$ for which code $(\cF,\cG)$ results in a decoding error. 
Let $(\cF_m,\cG_m)$ be the code obtained by executing $(\cF,\cG)$ in parallel $m$ times (with independent source information).
Namely, $(\cF_m,\cG_m)$ executes $m$ independent sessions of the original $(\e,\bR-\Delta,n,N)$ feasible-code $(\cF,\cG)$ on $m$ independent sub-messages.
For $i=1,\dots,k$, let $\tilde{W}^m_i=\tilde{W}_{i1},\dots,\tilde{W}_{im}$ be the messages corresponding to source $s_i$ in $(\cF_m,\cG_m)$ with $\tilde{W}_{ij} \in [2^{(R_i-\Delta)n}]$ for $j=1,\dots, m$.
For $i=1,\dots,k$, let $W^m_i \in [2^{(R_i-\Delta)n}]^{m(1-\delta)}$ for a parameter $\delta>0$ to be specified later.
Here, $W^m_i$ represents the message corresponding to source $i$ in an $(\e',\bR-\Delta',nm,N)$-feasible code $(\cF_\sigma,\cG_\sigma)$ that we will construct shortly.
For $i=1,\dots,k$, let $E_i: [2^{(R_i-\Delta)n}]^{(1-\delta)m} \rightarrow [2^{(R_i-\Delta)n}]^{m}$ be the encoder of an error correcting code over alphabet $[2^{(R_i-\Delta)n}]$ of rate $(1-\delta)$ and relative distance $H^{-1}(\delta)$. Here $H$ is the binary entropy function, and the existence of such a code follows from the Gilbert-Varshamov bound \cite{Gil52,Var57}.
We use code $E_i$ to map message $W_i^m \in [2^{(R_i-\Delta)n}]^{(1-\delta)m}$ to message $\tilde{W}_i^m \in [2^{(R_i-\Delta)n}]^{m}$. 
Denote by $E_{ij}$ the restriction of $E_i$ to the $j$'th entry of $E_i$ (over the alphabet $[2^{(R_i-\Delta)n}]$).
We use code $E_{ij}$ to map message $W_i^m$ to message $\tilde{W}_{ij}$. 
For $i=1,\dots,k$ and $j=1,\dots,m$, consider permutations $\sigma_{ij} : [2^{(R_i-\Delta)n}] \rightarrow [2^{(R_i-\Delta)n}]$ chosen uniformly and independently at random. 
After the error correcting code, we apply $\sigma_{ij}$ to message $\tilde{W}_{ij}$, permuting the symbol before transmission.

By the definitions above, it holds for any $j=1,\dots,m$ and any message $(W^m_1,\dots, W^m_k)$ that 
$$
\Pr_{(\sigma_{1j},\dots,\sigma_{kj})}[(\sigma_{1j}(E_{1j}(W^m_1)), \dots, \sigma_{kj}(E_{kj}(W^m_k))) \in A] \leq \e
$$
As the permutations for different $j=1,\dots,m$ are chosen independently, we can apply the Chernoff bound to conclude that for any fixed vector of messages $(W^m_1,\dots, W^m_k)$ and uniform and independent $\{\sigma_{ij}\}$, the probability that there are more than $2\e m$ values of $j$ between $1$ and $m$ for which 
$$
(\sigma_{1j}(E_{1j}(W^m_1)), \dots, \sigma_{kj}(E_{kj}(W^m_k))) \in A
$$
is at most $2^{-\e m/2}$.
This now implies the existence of permutations $\{\sigma_{ij}\}$ for which the probability over uniform and independent messages $(W^m_1,\dots, W^m_k)$ that there are more than $2\e m$ values of $j$ between $1$ and $m$ for which 
$$
(\sigma_{1j}(E_{1j}(W^m_1)), \dots, \sigma_{kj}(E_{kj}(W^m_k))) \in A
$$
is at most $2^{-\e m/2}$.

Let $\Delta'>0$.
Using the discussion above, we now describe an $(\e',\bR-\Delta',nm,N)$-feasible coding scheme $(\cF_\sigma,\cG_\sigma)$ for $\cI$ with $\e' = 2^{-\e m/2} \leq 1/(nm)^2$ for sufficiently large $m$ as follows.  
Roughly speaking, for encoding, $(\cF_\sigma,\cG_\sigma)$ consists of a pre-communication processing phase done at each source after which code $(\cF_m,\cG_m)$ is executed. 
Similarly, for decoding, terminals in $(\cF_\sigma,\cG_\sigma)$ first decode using code $(\cF_m,\cG_m)$ and then apply a post-communication processing phase.
Let $(W^m_1,\dots, W^m_k)$ be $k$ source messages for $(\cF_\sigma,\cG_\sigma)$. 
For each $j=1,\dots, m$, every source $i$ computes $\sigma_{ij}(E_{ij}(W^m_i))$, which is the input to be transmitted during the $j$'th session of $(\cF,\cG)$ in code $(\cF_m,\cG_m)$. 
For decoding, each terminal first applies the decoding of code $(\cF,\cG)$ to each of the $m$ sessions of $(\cF,\cG)$ in $(\cF_m,\cG_m)$. 
Since the probability that  there are more than $2\e m$ values of $j$ between $1$ and $m$ for which $(\sigma_{1j}(E_{1j}(W^m_1)), \dots, \sigma_{kj}(E_{kj}(W^m_k))) \in A$ is 
at most $\e'$; it holds that with probability at least $1-\e'$ over $(W^m_1,\dots, W^m_k)$ each terminal will correctly decode all but $2\e m$ of $(\sigma_{i1}(E_{i1}(W^m_i)), \dots, \sigma_{im}(E_{im}(W^m_i)))$ for each source $i$ it required.
Reversing the permutations $\sigma_{i1},\dots,\sigma_{im}$ the terminal can recover a reconstruction of $E_i(W^m_i)=(E_{i1}(W^m_i), \dots, E_{im}(W^m_i))$ in which at most $2\e m$ entries $E_{ij}(W^m_i)$ are corrupted. Applying a nearest codeword decoding recovers $W^m_i$ as long as the minimum distance of the code $E_i$ is at least $4\e m +1$. Thus setting $\delta$ to satisfy $H^{-1}(\delta) = 4\e + \frac{1}{m}$ implies correct decoding.
The resulting code $(\cF_\sigma,\cG_\sigma)$ is $(\e',(\bR-\Delta)(1-\delta),nm,N)$-feasible for $\cI$. 
Finally, starting our analysis with sufficiently small $\e>0$ and $\Delta>0$, we conclude that the $(\e',(\bR-\Delta)(1-\delta),nm,N)$-feasible code $(\cF_\sigma,\cG_\sigma)$ for $\cI$ is also $(\e',\bR-\Delta',nm,N)$-feasible.
This shows that we can use a code with inner-blocklength $n$ and error probability $\e$ to build a code with inner-blocklength $nm$ and error probability $\e'$.

We next seek to present an $(\e',\bR-2\Delta',n+1,Nm)$-feasible code for $\cI$ with $\e'  \leq 1/(Nm)^2$ for sufficiently large $m$.  We start with the $(\e',\bR-\Delta',nm,N)$-feasible code $(\cF_\sigma,\cG_\sigma)$ for $\cI$ discussed above. Observing that the alphabet $[2^{\lambda nm}]$ is included in the alphabet $[2^{\lambda (n+1)}]^m$ for sufficiently large $n$, we conclude that $(\cF_\sigma,\cG_\sigma)$ is also $(\e',\frac{n}{n+1}(\bR-\Delta'),n+1,Nm)$-feasible, which in turn, for sufficiently large $n$ is $(\e',\bR-2\Delta',n+1,Nm)$-feasible as required.
\end{proof}

\begin{lemma}
\label{lem:alpha}
Let $\cI=(G,S,D,M)$. Let $\alpha>0$, and define $\alpha\cI=(\alpha G,S,D,M)$ to be the instance obtained by multiplying each edge capacity in $G$ by $\alpha$ (to obtain a graph here described as $\alpha G$). Then $\alpha\cR(\cI)=\cR(\alpha\cI)$ and $\alpha\cR_0(\cI)=\cR_0(\alpha\cI)$. 
\end{lemma}
\begin{proof}
We first note that we only need to prove, e.g., that $\cR(\alpha\cI) \subseteq \alpha\cR(\cI)$ as the other direction
$\cR(\alpha\cI) \supseteq \alpha\cR(\cI)$ then follows from taking $\cI'=\alpha\cI$ and $\alpha'=\frac{1}{\alpha}$ to obtain $\frac{1}{\alpha}\cR(\alpha\cI)=\alpha'\cR(\cI') \supseteq \cR(\alpha'\cI')=\cR(\cI)$.

For the direction $\cR(\alpha\cI) \subseteq \alpha\cR(\cI)$, consider $\bR \in \cR(\alpha\cI)$. Let $\e>0$ and $\Delta>0$. Then, $\bR \in \cR(\alpha\cI)$ implies the existence of an $(\e,\bR-\Delta,n,N)$ feasible code $({\mathcal F},\mathcal{G})$ for $\alpha \cI$ for infinitely many values of $n$ and $N$. We now argue that any code of outer-blocklength $N$ and inner-blocklength $n$ over $\alpha \cI$ can be executed on $\cI$ by a code of outer-blocklength $N$ and inner-blocklength $\lceil \alpha n \rceil$.
Specifically, each time step in $\alpha \cI$ over an inner-blocklength $n$ is executed by a single time step in $\cI$ over an inner-blocklength of $\lceil \alpha n \rceil$. 
We conclude an $(\e,\frac{n}{\lceil \alpha n \rceil}(\bR-\Delta),\lceil \alpha n \rceil,N)$ code for $\cI$ for infinitely many values of $n$ and $N$. 
As $\e>0$ and $\Delta>0$ are arbitrary, and $n$ can be taken to be arbitrarily large, we have that $\frac{\bR}{\alpha} \in \cR(\cI)$, or equivalently, $\bR \in \alpha\cR(\cI)$.
\end{proof}

\begin{figure*}[t!]
\centering
\includegraphics[scale=0.6]{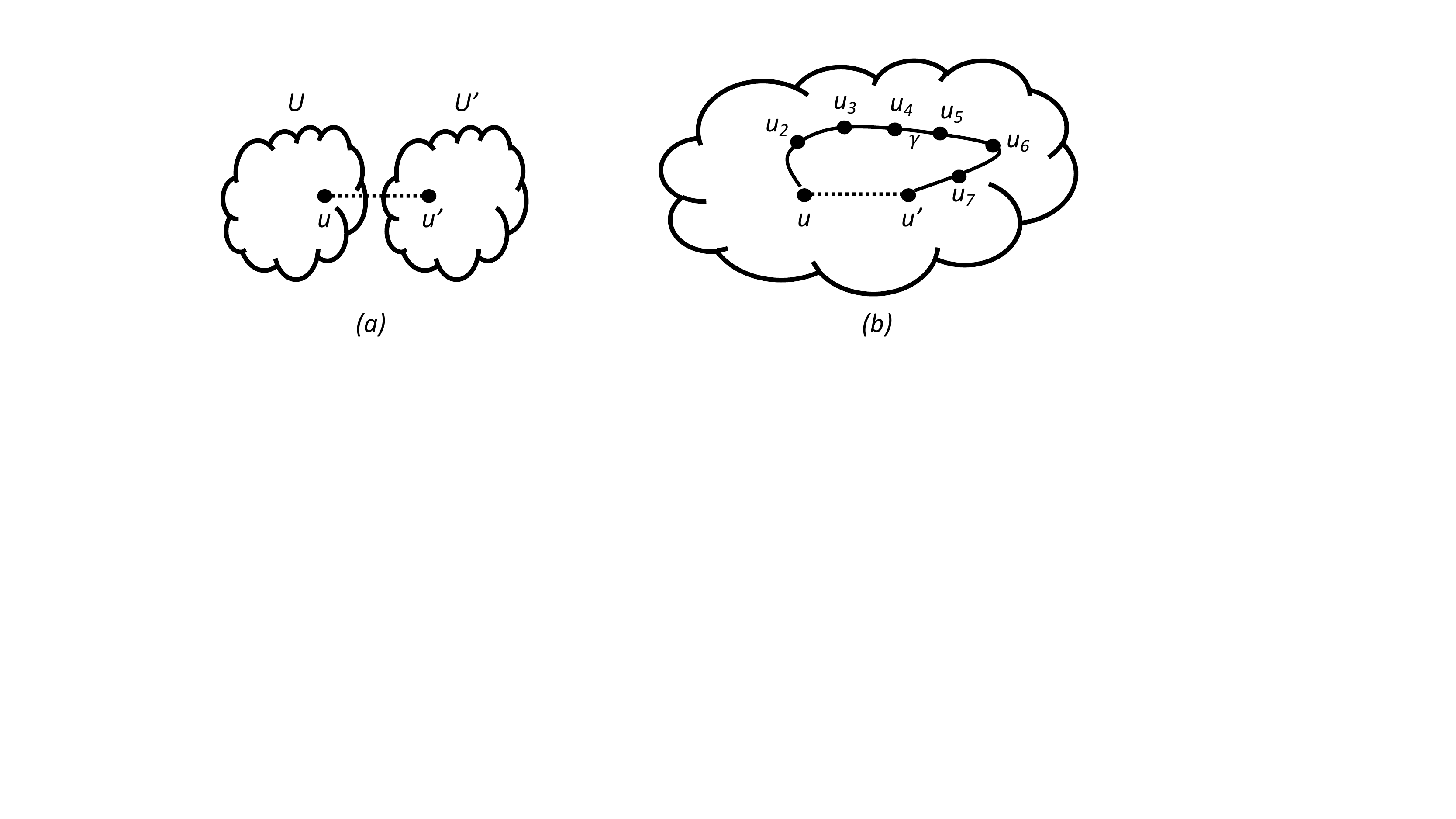}
\vspace{-70mm}
\caption{A schematic description of the items in Theorem~\ref{the:main}. The edge $e=(u,u')$ is marked as a dotted line. The first case (a) in which the graph $G$ of instance $\cI$ is disconnected and adding the edge $e=(u,u')$ connects between two components $U$ and $U'$ of $G$. The second case (b) in which there exists in $G$ a path connecting $u$ and $u'$. Here, the path is $u,u_2,u_3,\dots,u_{7},u'$ and $\gamma$ is the value of the minimum capacity edge $(u_4,u_5)$ in the path.}
\label{fig:cases}
\end{figure*}

%

\subsection{Edge Removal}
Throughout the discussions in this work, we use the term ``edge-removal statement,'' often shortened to ``edge removal,'' to refer to the mathematical statement defined here.
{Let $\mathcal{I}=(G,S,D,M)$.}
Let $\Gl$ be the graph obtained by adding an (undirected) edge $e$ of capacity $\lambda >0$ to $G$.\footnote{{Instead of starting with a network and then {\em removing} an edge as in~\cite{HEJ:10,JEH:11}, 
it is more convenient for our presentation to start with a network and then {\em add} an edge as in \cite{LE:19}.}}
Let $\Il=(\Gl,S,D,M)$ describe the resulting network instance. 
The edge-removal statements given below compare the rate vectors  achievable over $\cI$ and $\Il$.
We use notation stemming from \cite{LE:19} to define the following variants of edge removal.

\begin{statement}[The edge-removal statements on instance $\cI$]
\label{state:edge:f}
$\bullet\ $The {\em edge-removal statement} holds with function $\fe$ on instance $\cI$ and edge capacity $\lambda$ if for any edge $e \in V \times V$ 
$$ \bR \in \cR(\Il) \Longrightarrow \bR-\fe \in \cR(\cI).$$
$\bullet\ $The {\em zero-error edge-removal statement}  holds with function $\fe$ on instance $\cI$ and edge capacity $\lambda$ if for any edge $e \in V \times V$ 
$$ \bR \in \cR_0(\Il) \Longrightarrow \bR-\fe \in \cR_0(\cI).$$
$\bullet\ $The {\em vanishing-edge-removal statement} holds on instance $\cI$ if for any edge $e \in V \times V$ 
$$\cR(\cI)=\lim_{\lambda \rightarrow 0} \cR(\Il). $$
$\bullet\ $The {\em zero-error vanishing-edge-removal statement} holds on instance $\cI$ if for any edge $e \in V \times V$ 
$$\cR_0(\cI)=\lim_{\lambda \rightarrow 0} \cR_0(\Il). $$
\end{statement}

\section{Main results}
\label{sec:main}

We now present the main results of this work. 

\begin{theorem}
\label{the:main}
Let $\cI=(G,S,D,M)$ be an undirected network instance. Let $\lambda >0$. The edge-removal statement and the  zero-error edge-removal statement hold with function $\fe = c{\lambda}$ for some constant $c$ that depends only on the edge capacities in $\cI$.
\end{theorem}

The proof of Theorem~\ref{the:main} applies the follow lemma.

\begin{lemma}
\label{lem:path}
Let $\cI=(G,S,D,M)$. Let $e=(u,u')$ be an edge in $G$ of capacity $\lambda$. Let $\cI_{path}=(G_{path},S,T,M)$ be the instance obtained by modifying graph $G$ to yield a new graph $G_{path}$ as follows.
Starting with $\cI$, we wish to remove the edge $e=(u,u')$ and replace it with a path $u=u_1,u_2,u_3, \dots, u_{\ell-1},u_\ell=u'$ of length $\ell-1$ and capacity $\lambda$, where $u_i$ for $i = 2,3,\dots,{\ell-1}$ are nodes in $G$. When compared to $G$, the capacity of path edges $(u_i,u_{i+1})$ is increased in $G_{path}$ by $\lambda$. Then $\cR(\cI) \subseteq \cR(\cI_{path})$ and $\cR_0(\cI) \subseteq \cR_0(\cI_{path})$.
\end{lemma}

We start by proving Theorem~\ref{the:main} using Lemma~\ref{lem:path} above. We then prove Lemma~\ref{lem:path}. 

\begin{proof}
Let $\cI = (G,S,D,M)$. For any nodes $u$ and $u'$ in $V$, let $e=(u,u')$ be an edge of capacity $\lambda$ to be added to $G$. 
Let $\Il = (G^{\lambda,e},S,D,M)$.
Let $\bR^{\lambda,e}=(R_1^{\lambda,e},\dots,R_k^{\lambda,e}) \in \cR(\cI^{\lambda,e})$.
We show that $\bR^{\lambda,e}-c\lambda \in \cR(\cI)$ for some constant $c$ that depends only on the edge capacities in $G$.
We consider two cases depicted in Figure~\ref{fig:cases}. 

In the first case, we assume that the graph $G$ of instance $\cI$ is disconnected and that adding the edge $e=(u,u')$ connects two unconnected components of $G$. Let $U$ and $U'$ be a partition of the vertex set $V$ of $G$ such that $u \in U$, $u' \in U'$, and subsets $U$ and $U'$ are disconnected in $G$. In this case, the added edge $e$ acts as a bridge of capacity $\lambda$ between (perhaps subsets of) $U$ and $U'$.

We first consider sources $s_i \in U$ for which there exist  terminals $d_j$ in $U'$ such that $m_{ij}=1$. As $e$ is a bridge between $U$ and $U'$, then using the cut set bound (e.g., Corollary 25 of \cite{harvey2006capacity}) it follows that $R_i^{\lambda,e} \leq \lambda$. The same holds if $s_i \in U'$ and $d_j \in U$. 

We now consider all sources $s_i \in U$ such that all $d_j$ for which $m_{ij}=1$ satisfy $d_j \in U$. Denote this set of sources by $S_U$. Let $D_U = D \cap U$ be the set of terminals in $U$. Let $G_U$ be the subgraph of $G$ induced by the vertices in $U$. Finally, let $M_U$ be the minor of $M$ induced on columns and rows of $M$ corresponding to $U$. Consider the instance $\cI_U=(G_U,S_U,D_U,M_U)$. We now claim that $\bR_U=(R_i^{\lambda,e}:i \in S_U) \in \cR(\cI_U)$. 
We prove this claim using the following averaging argument. Let $({\mathcal F},\mathcal{G})$ be an 
$(\e,\bR^{\lambda,e}-\Delta,n,N)$-feasible network code for $\Il$. 
By an averaging argument on the source messages $\bW=(W_1,\dots,W_k)$, there exist fixed values $(w_i:i \not \in S_U)$ for $(W_i:i \not \in S_U)$ for which the probability of successful communication using $({\mathcal F},\mathcal{G})$ conditioned on  $(W_i=w_i:i \not \in S_U)$ does not exceed $\e$. Moreover, under the condition $(W_i=w_i:i \not \in S_U)$ the network code  $({\mathcal F},\mathcal{G})$ can be {\em simulated} on $\cI_U$. 
That is, there exists a network code $({\mathcal F}_U,\mathcal{G}_U)$ for $\cI_U$ in which, for any time step $t \in [N]$ and for any edge $e' \in G_U$, the values transmitted over $e'$ in $\Il$ using $({\mathcal F},\mathcal{G})$ are also transmitted over edge $e'$ in $\cI_U$ using $({\mathcal F}_U,\mathcal{G}_U)$. This follows from the fact that node $u$ can simulate all incoming information from node $u'$ in $({\mathcal F},\mathcal{G})$ given the knowledge that $(W_i:i \not \in S_U)=(w_i:i \not \in S_U)$. Thus, $({\mathcal F}_U,\mathcal{G}_U)$ is an $(\e,\bR_U-\Delta,n,N)$-feasible network code for $\cI_U$. As the argument applies for any $\e>0$ and $\Delta>0$, this implies that $\bR_U \in \cR(\cI_U)$.
Similarly, one can define $\cI_{U'}$ and show that $\bR_{U'}=(R_i^{\lambda,e}:i \in S_{U'}) \in \cR(\cI_{U'})$. Thus, we conclude that the rate vector $\bR = (R_1,\dots,R_k)$ for which $R_i = R_i^{\lambda,e}$ if $i \in S_U \cup S_{U'}$ and $R_i = 0$ otherwise satisfies $\bR \in \cR(\cI)$. This follows by running codes over $\cI_U$ and $\cI_{U'}$ in parallel. As, for each $i=1,\dots,k$, the above analysis implies that $R_i \geq R_i^{\lambda,e} - \lambda$ we conclude the assertion of the theorem for the case under study with $\fe = \lambda$.

In the second case, we assume that there exists in $G$ a path connecting $u$ and $u'$. Let $u=u_1,u_2,u_3,u_{\ell-1},u'=u_{\ell}$ be one such path and let $\gamma$ be the capacity of the minimum capacity edge in the path. Let $\delta > 0$ satisfy $\lambda = \delta\gamma$. Consider the graph $G_{path}$ obtained from $G^{\lambda,e}$ by removing the edge $e=(u,u')$ and increasing the capacity of all edges in the path $u_1,u_2,u_3,u_{\ell-1},u_{\ell}$ by $\lambda$. Let $\cI_{path}$ be the instance $(G_{path},S,D,M)$. 
By Lemma~\ref{lem:path}, if $\bR^{\lambda,e}=(R_1^{\lambda,e},\dots,R_k^{\lambda,e}) \in \cR(\cI^{\lambda,e})$ then $\bR^{\lambda,e} \in \cR(\cI_{path})$. 
Let $\alpha=\frac{\gamma}{\gamma+\lambda} = \frac{1}{1+\delta}$. Consider the instance $\alpha \cI_{path} = (\alpha G_{path},S,D,M)$. 
By Lemma~\ref{lem:alpha}, $\bR^{\lambda,e}\in \cR(\cI^{\lambda,e})$  implies $\alpha\bR^{\lambda,e} \in \cR(\alpha\cI_{path})$. 
Notice that the capacity of every edge in $\alpha G_{path}$ is at most the capacity of the corresponding edge in $G$. This is clearly true for edges in $G_{path}$ that are not on the path $u_1,u_2,u_3,u_{\ell-1},u_{\ell}$, and holds for path-edge $(u_i,u_{i+1})$ of capacity $\gamma' \geq \gamma$ as $\alpha(\gamma'+\lambda) \leq \gamma'$. Thus, $\alpha\bR^{\lambda,e} \in \cR(\alpha\cI_{path}) \subseteq \cR(\cI)$. 

Notice that 
$\bR^{\lambda,e} - \fe \leq \alpha\bR^{\lambda,e}$, and thus $\bR^{\lambda,e} - \fe \in \cR(\cI)$ for 
$$\fe = \frac{\delta}{1+\delta}\max_i{R_i^{\lambda,e}} \leq \delta(W(\cI)+\lambda) \leq 2\delta W(\cI)$$
where $W(\cI)$ is the sum of all edge capacities in $G$ and $(W(\cI)+\lambda) \leq 2 W(\cI)$ is the sum of all edge capacities in $G^{\lambda,e}$. 
Here, we assume without loss of generality that $\lambda \leq W(\cI)$ (as otherwise for all $i$, $R^{\lambda,e} \leq 2\lambda$, a setting in which the proof of the theorem is immediate). 
Thus, $\fe \leq 2\delta W(\cI) = \frac{2W(\cI)}{\gamma}\lambda$. Let $w(\cI)$ be the minimum edge capacity over all edges in $G$, then 
$\fe \leq \frac{2W(\cI)}{w(\cI)}\lambda = c\lambda$ for a constant $c=\frac{2W(\cI)}{w(\cI)}$ that only depends on the capacities of edges in $\cI$. 

An identical proof holds for the zero-error case.
\end{proof}

We now prove Lemma~\ref{lem:path}.

\noindent
\begin{proof}
The proof follows the line of proof given in \cite{effros2012dependence,koetter2011theory}, in which it is shown that adding constant delays (independent of the blocklegth) in network communication has no impact on capacity.  Throughout, to simplify our presentation, we  consider the instance $\cI^*_{path}=(G^*_{path},S,D,M)$ in which we define $G^*_{path}$ (similar to $G_{path}$) by starting with $G$, removing the edge $e$, and replacing it with a path of length $\ell-1$ and capacity $\lambda$ consisting of nodes $u=u_1,u_2,u_3, \dots, u_{\ell-1},u_\ell=u'$. However, in $G^*_{path}$ the nodes $u_i$ for $i = 2,3,\dots,{\ell-1}$ are new nodes that do not originally appear in $G$. 
As any network code, for sources $S$, terminals $D$, and demands $M$, on $\cI^*_{path}$ can be implemented on $\cI_{path}$, it holds that $\cR(\cI^*_{path}) \subseteq \cR(\cI_{path})$. To conclude our proof, we show that $\cR(\cI) \subseteq \cR(\cI^*_{path})$.

The proof that $\cR(\cI) \subseteq \cR(\cI^*_{path})$ proceeds as follows. 
For any $\bR \in \cR(\cI)$ and $\Delta>0$, by Lemma~\ref{lemma:eps}, consider an $(\e,\bR-\Delta,n,N)$ feasible code $(\cF,\cG)$ on $\cI$ with $n$ and $N$ sufficiently large and with $\e \leq 1/N^2$. 
{\em Interleaving} $N$ such codes on $N$ independent sub-messages from each source in $S$, as in \cite{effros2012dependence}, we obtain a new code $(\tilde{\cF},\tilde{\cG})$ that is $(N\e,\bR-\Delta,n,N^2)$ feasible as follows. 
The new code executes $N$ independent sessions of the original $(\e,\bR-\Delta,n,N)$ feasible-code $(\cF,\cG)$ on $N$ independent sub-messages.
The sessions operate in a time-interleaved manner. 
In time steps $t=1$ through $t=N$ of $(\tilde{\cF},\tilde{\cG})$, the first time step of all independent sessions of $(\cF,\cG)$ is executed. (Time step $1$ of independent session $j$ of $(\cF,\cG)$ operates in time step $j$ of $(\tilde{\cF},\tilde{\cG})$.) 
In general, in time steps $t=(i-1)N+1$ through $t=iN$ of $(\tilde{\cF},\tilde{\cG})$, the $i$'th time step of each independent session of $(\cF,\cG)$ is executed. 
(Time step $i$ of independent session $j$ of $(\cF,\cG)$ operates in time step $(i-1)N+j$ of $(\tilde{\cF},\tilde{\cG})$.)
After $N^2$ time steps, the $N$ independent sessions of $(\cF,\cG)$ are completed, implying an $(N\e,\bR-\Delta,n,N^2)$ feasible code $(\tilde{\cF},\tilde{\cG})$ for $\cI$. 
Here, we bound the error by a union bound over the individual independent sessions of $(\cF,\cG)$.
The given $(N\e,\bR-\Delta,n,N^2)$ feasible code $(\tilde{\cF},\tilde{\cG})$ for $\cI$ satisfies 
$\overleftarrow{\cX^n_{e,t}} = \overleftarrow{\cX^n_{e,t'}}$ for any $t=(i-1)N+j$ and $t'=(i-1)N+j'$ with $j,j'\in[N]$.
That is, for any time steps $t$ and $t'$ of $(\tilde{\cF},\tilde{\cG})$ in the same {\em sub-block} of length $N$, we have $\overleftarrow{\cX^n_{e,t}} = \overleftarrow{\cX^n_{e,t'}}$, as in that sub-block we are executing $N$ independent session of the same time step $i$ in the original code $(\cF,\cG)$.
Similarly, $\overrightarrow{\cX^n_{e,t}}=\overrightarrow{\cX^n_{e,t'}}$ for all $t,t' \in [(i-1)N+1, iN]$. 
Here and in what follows, we refer to time steps $t=(i-1)N+j$ for $j \in [N]$ in $(\tilde{\cF},\tilde{\cG})$ as the $i$'th sub-block of time steps of code $(\tilde{\cF},\tilde{\cG})$.
This will be useful when we consider $\cI^*_{path}$.

\begin{figure*}[t!]
\centering
\includegraphics[scale=0.46]{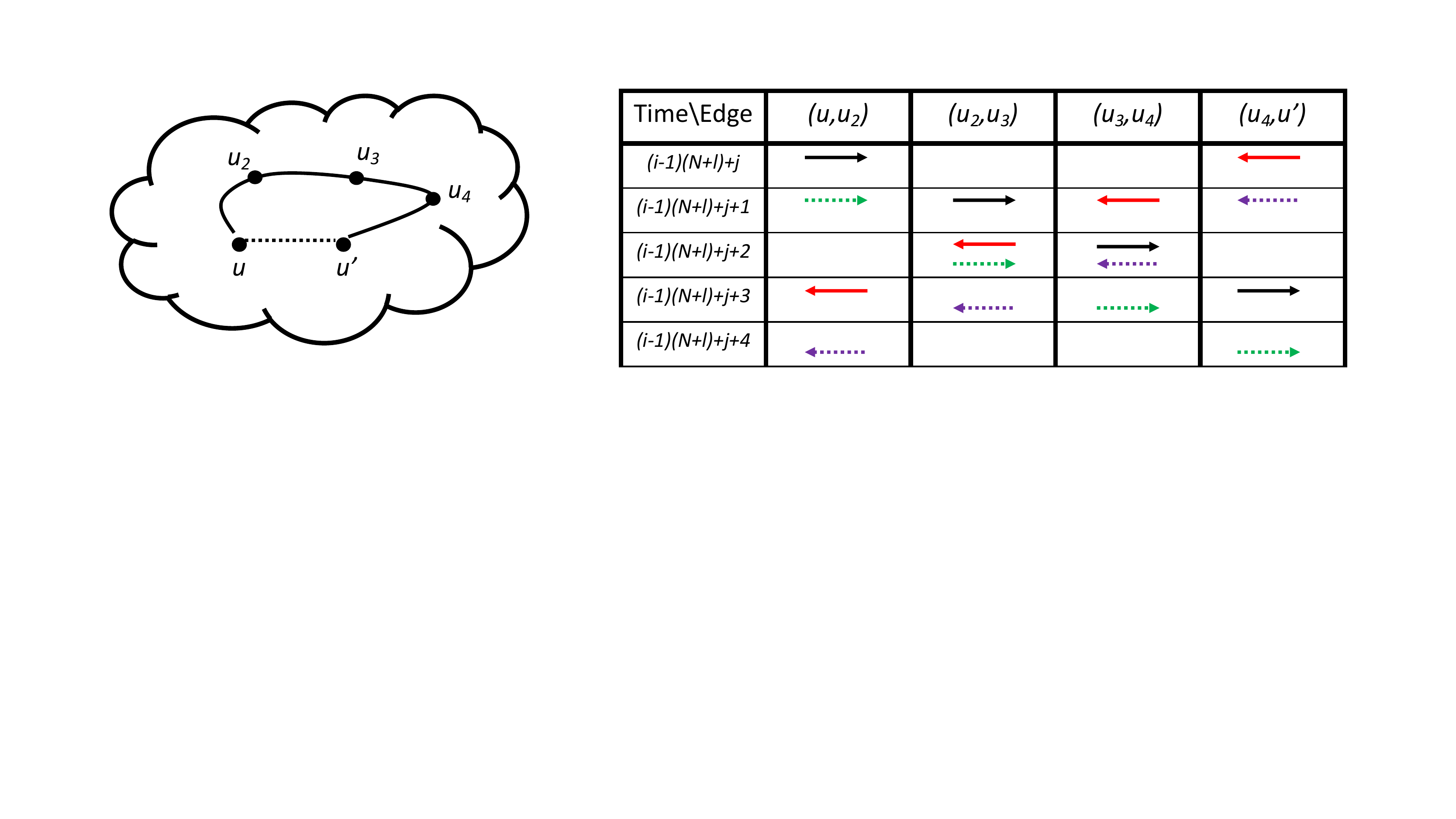}
\vspace{-45mm}
\caption{An illustration of how information sent across edge $e=(u,u')$ in $(\tilde{\cF},\tilde{\cG})$ traverses the path $u=u_1,u_2,u_3,u_4,u_5=u'$ in $(\cF_{path},\cG_{path})$.  In this example, $\ell=5$. The solid right arrows (in black) represent information sent from $u$ to $u'$ at time $t=(i-1)(N+\ell)+j$ in $(\tilde{\cF},\tilde{\cG})$. The dotted right arrows (in green) represent information sent from $u$ to $u'$ at time $t=(i-1)(N+\ell)+j+1$ in $(\tilde{\cF},\tilde{\cG})$.  The solid left arrows (in red) represent information sent from $u'$ to $u$ at time $t=(i-1)(N+\ell)+j$ in $(\tilde{\cF},\tilde{\cG})$. The dotted left arrows (in purple) represent information sent from $u'$ to $u$ at time $t=(i-1)(N+\ell)+j+1$ in $(\tilde{\cF},\tilde{\cG})$. 
}
\label{fig:pipeline}
\end{figure*}

We now use the $(N\e,\bR-\Delta,n,N^2)$ code $(\tilde{\cF},\tilde{\cG})$ to construct an  $(N\e,\bR-\Delta,n,N(N+\ell))$ code for $\cI^*_{path}$. That is, we use a code for the network that includes edge $e=(u,u')$ to build a code for the network in which edge $e=(u,u')$ is removed and replaced with a path of length $\ell-1$ of the same capacity. 
As $N$ is chosen above to be ``sufficiently large'', we assume here that $\ell$ is significantly smaller than $N$, say $\ell = \delta N$ for any constant $\delta > 0$ of our choice.
The code $(\cF_{path},\cG_{path})$ for $\cI^*_{path}$ is a slight modification of the code $(\tilde{\cF},\tilde{\cG})$.
In $(\cF_{path},\cG_{path})$ we still communicate with inner-blocklength $n$.
However the outer-blocklegth is set to $N(N+\ell)$ instead of $N^2$. 
Like code $(\tilde{\cF},\tilde{\cG})$, code $(\cF_{path},\cG_{path})$ operates in sub-blocks, where for each $i \in [N]$, sub-block $i$ here takes the form $t=(i-1)(N+\ell)+j$ for $j \in [N+\ell]$.
Roughly speaking, the $i$'th sub-block of code $(\tilde{\cF},\tilde{\cG})$ determines the $i$'th sub-block of code $(\cF_{path},\cG_{path})$. 
The first $N$ time steps in each sub-block $i$ of $(\cF_{path},\cG_{path})$ 
perform precisely the same operations as are performed in sub-block $i$ of
$(\tilde{\cF},\tilde{\cG})$. 
The last $\ell$ time steps in each sub-block in $(\cF_{path},\cG_{path})$ are used 
to transmit the information sent across edge $e=(u,u')$ in $(\tilde{\cF},\tilde{\cG})$ along the path of length $\ell-1$ that replaced edge $e$ in $(\cF_{path},\cG_{path})$.

We first describe the coding operations of  $(\cF_{path},\cG_{path})$ on edges $e'$ that are not on the path $u=u_1,u_2,u_3, \dots, u_{\ell-1},u_\ell=u'$. 
For any such edge $e'=(v,v')$, the transmitted message from $v$ to $v'$ at time step $t=(i-1)(N+\ell)+j$ in $(\cF_{path},\cG_{path})$ equals the transmitted message from $v$ to $v'$ in $(\tilde{\cF},\tilde{\cG})$ in time step $t=(i-1)N+j$. That is, the message over $e'$ in the $j$'th time step of  the $i$'th sub-block in $(\cF_{path},\cG_{path})$ equals the message over $e'$ in the $j$'th time step of  the $i$'th sub-block in $(\tilde{\cF},\tilde{\cG})$. 
In the remaining time steps in each sub-block of $(\cF_{path},\cG_{path})$, no information is transmitted in either direction over edge $e'$. 
Namely, for each sub-block $i$, in time steps $t=(i-1)(N+\ell)+(N+1)$ through $t=i(N+\ell)$ a predetermined fixed message is transmitted over $e'$.

We now describe the coding operations of  $(\cF_{path},\cG_{path})$ on edges $e'=(u_r,u_{r+1})$ on the path $u=u_1,u_2,u_3, \dots, u_{\ell-1},u_\ell=u'$. 
Roughly speaking, these edges ``pipe-line" the message transmitted  over the removed edge $e=(u,u')$ from $u$ to $u'$ and from $u'$ to $u$, in $G$. (See Figure~\ref{fig:pipeline} for an illustration.) For edge $e'=(u_r,u_{r+1})$, the transmitted message from $u_r$ to $u_{r+1}$ in time step $t=(i-1)(N+\ell)+j+r-1$ in $(\cF_{path},\cG_{path})$ equals the transmitted message from $u$ to $u'$ over $e=(u,u')$ in $(\tilde{\cF},\tilde{\cG})$ in time step $t=(i-1)N+j$. 
In addition, the transmitted message from $u_{r+1}$ to $u_{r}$ in time step $t=(i-1)(N+\ell)+j+\ell-r-1$ of $(\cF_{path},\cG_{path})$ equals the transmitted message from $u'$ to $u$ over $e=(u',u)$ in $(\tilde{\cF},\tilde{\cG})$ in time step $t=(i-1)N+j$. 
That is, the message over $e=(u,u')$ from $u$ to $u'$ in the $j$'th time step of the $i$'th block of $(\tilde{\cF},\tilde{\cG})$ traverses the path $u=u_1,u_2,u_3, \dots, u_{\ell-1},u_\ell=u'$ in time steps $t=(i-1)(N+\ell)+j$ through $t=(i-1)(N+\ell)+j+\ell-2$ in $(\cF_{path},\cG_{path})$. 
In the other direction, the message over $e=(u',u)$ from $u'$ to $u$ in the $j$'th time step of the $i$'th block of $(\tilde{\cF},\tilde{\cG})$ traverses the path $u'=u_\ell,u_{\ell-1}, \dots, u_{2},u_1=u$ in time steps $t=(i-1)(N+\ell)+j$ through $t=(i-1)(N+\ell)+j+\ell-2$ in $(\cF_{path},\cG_{path})$. 


We now show that the above communication scheme of $(\cF_{path},\cG_{path})$ is feasible on $\cI^*_{path}$. 
We first show that any message on edge $e'=(v,v')$ from $v$ to $v'$ in time step $t$ in $(\cF_{path},\cG_{path})$ can be computed in $G^*_{path}$ from the information available to node $v$ prior to time step $t$. 
If the edge $e=(u,u')$ removed from $\cI$ is not an incoming edge to $v$ in $\cI$ (i.e., if $v\ne u$ and $v \ne u'$), this follows directly by the feasibility of $(\tilde{\cF},\tilde{\cG})$ over $\cI$.
Otherwise, one must take into account the {\em delay} incurred be replacing $e$ with the path of length $\ell-1$. However, due to the interleaved structure of $(\tilde{\cF},\tilde{\cG})$, this delay  does not impact the feasibility of $(\cF_{path},\cG_{path})$. Specifically,  the structure of $(\tilde{\cF},\tilde{\cG})$ ensures that any message transmitted over an edge incoming to $v$ in sub-block $i$ of $(\tilde{\cF},\tilde{\cG})$ will be used by node $v$ as input to subsequent encoding only in the next sub-block $i+1$ of $(\tilde{\cF},\tilde{\cG})$. 
The same holds for $(\cF_{path},\cG_{path})$. Moreover, by our definitions, any message transmitted over an edge incoming to $v$ in sub-block $i$ of $(\tilde{\cF},\tilde{\cG})$ must also be transmitted to $v$ in sub-block $i$ of $(\cF_{path},\cG_{path})$.  We conclude that any message on edge $e'=(v,v')$ from $v$ to $v'$ in time step $t$ in $(\cF_{path},\cG_{path})$ can be computed in $G^*_{path}$ from the information available to $v$ prior to time step $t$.  

Secondly, we show that the alphabet $\cX^n_{e'}$ of edges $e'=(v,v')$ in $G^*_{path}$ can support the code $(\cF_{path},\cG_{path})$ in the sense that the messages transmitted from $v$ to $v'$ and from $v'$ to $v$ in time step $t$ have support $\overrightarrow{\cX^n_{e',t}}$ and $\overleftarrow{\cX^n_{e',t}}$ respectively, that satisfy
$|\overrightarrow{\cX^n_{e',t}}| \cdot |\overleftarrow{\cX^n_{e',t}}| \leq |\cX^n_{e'}|$.
For edges $e'$ that are not on the path $u=u_1,u_2,u_3, \dots, u_{\ell-1},u_\ell=u'$, this follows directly by the  feasibility of the code $(\tilde{\cF},\tilde{\cG})$. For edges $e'=(u_r,u_{r+1})$ along the path, the alphabet $\cX^n_{e'}$ corresponding to $e'$ equals $\cX^n_e$ of (the removed edge) $e=(u,u')$ in $G$. 
Moreover, as discussed above, we have that $\overleftarrow{\cX^n_{e,t}} = \overleftarrow{\cX^n_{e,t'}}$ and $\overrightarrow{\cX^n_{e,t}} = \overrightarrow{\cX^n _{e,t'}}$ for any $t$ and $t'$ in the $i$'th sub-block of  $(\tilde{\cF},\tilde{\cG})$.
We thus denote the alphabets $\overleftarrow{\cX^n_{e,t}}$ and $\overrightarrow{\cX^n_{e,t}}$ in $(\tilde{\cF},\tilde{\cG})$ for any $t$ in the $i$'th sub-block of $(\tilde{\cF},\tilde{\cG})$ (i.e., $t=(i-1)N+j$ for $j \in [N]$) by $\overleftarrow{\cX^n_{e,{\tt block \ i}}}$ and $\overrightarrow{\cX^n_{e,{\tt block \ i}}}$ respectively.
Thus, for any $t$ in the $i$'th sub-block of $(\cF_{path},\cG_{path})$, i.e. $t=(i-1)(N+\ell)+j$ for $j \in [N+\ell]$, and any edge $e'=(u_r,u_{r+1})$ on the path $u=u_1,u_2,u_3, \dots, u_{\ell-1},u_\ell=u'$, we define alphabets $\overleftarrow{\cX^n_{e',t}}$ and $\overrightarrow{\cX^n_{e',t}}$ in $(\cF_{path},\cG_{path})$ to be equal to $\overleftarrow{\cX^n_{e,{\tt block \ i}}}$ and $\overrightarrow{\cX^n_{e,{\tt block \ i}}}$ of $(\tilde{\cF},\tilde{\cG})$ respectively. 
Such a definition for $\overleftarrow{\cX^n_{e',t}}$ and $\overrightarrow{\cX^n_{e',t}}$ allows $e'=(u_r,u_{r+1})$ to support the messages defined previously by $(\cF_{path},\cG_{path})$ during time steps $t=i(N+\ell)+j$ for $j \in [N+\ell]$. (See, for example, edge $(u_2,u_3)$ of Figure~\ref{fig:pipeline} in time step $t=(i-1)(N+\ell)+j+2$.)

We therefore conclude that $(\cF_{path},\cG_{path})$ is an $(N\e,\frac{N^2}{N(N+\ell)}(\bR-\Delta),n,N(N+\ell))$-feasible code for $\cI^*_{path}$. Finally, as $\e \leq 1/N^2$, $\Delta>0$ can be chosen to be arbitrarily small, and for any $\delta>0$ we can choose $N$ sufficiently large such that $\ell=\delta N$, we conclude that $\bR \in \cR(\cI^*_{path})$.

An almost identical proof (with very slight modifications) holds for the zero-error case as well.
\end{proof}

%
%

Corollary~\ref{cor:main} follows immediately from Theorem~\ref{the:main}. 
\begin{corollary}
\label{cor:main}
The vanishing edge-removal statement and the zero-error vanishing-edge-removal statement hold for undirected network instances. 
\end{corollary}

Using the connections outlined in \cite{LE:19}, Corrollary~\ref{cor:zero} also follows from Theorem~\ref{the:main}.
\begin{corollary}
\label{cor:zero}
Let $\cI$ be an undirected network instance, then $\cR(\cI)=\cR_0(\cI)$.
\end{corollary}
\begin{proof}
We use Theorem 3.1(a) of \cite{LE:19} to show in Corrollary~\ref{cor:zero} that $\cR(\cI) \subseteq \cR_0(\cI)$ (the other direction is immediate).
Namely,  in Theorem 3.1(a) of \cite{LE:19} it is shown, given an instance $\cI$ and $\bR \in \cR(\cI)$, how to construct an instance $\cI_1$ and an edge $e$ in $\cI_1$ such that
(i) $\cR_0(\cI)=\cR_0(\cI_1)$ and
(ii) for any $\lambda>0$, it holds that $\bR \in \cR_0(\cI_1^{\lambda,e})$.
By Theorem~\ref{the:main}, $\bR - c\lambda \in \cR_0(\cI_1)$ for a constant $c$ that depends only on the capacities of edges in $\cI_1$. This, in turn, implies that $\bR - c\lambda \in \cR_0(\cI)$. As $\cR_0(\cI)$ is, by definition, a closed set and $\lambda$ can be taken to be arbitrarily small, we conclude that $\bR \in \cR_0(\cI)$.
\end{proof}

\section{Conclusions}
In this work, we study the edge removal problem on undirected networks. Using the conceptually simple idea of {\em re-routing} information on the removed edge (if possible in the given topology) we show that the asymptotic version of the edge-removal statement holds. 
That is, we show that removing an edge of negligible capacity in undirected networks has only a negligible impact on the capacity region. 
This, in turn, implies that the zero-error capacity region of an undirected network equals its vanishing-error capacity region. Whether similar results are true for directed networks is an intriguing open problem. In addition, in light of the multiple-unicast coding advantage conjecture on undirected networks, it would be interesting to prove Theorem~\ref{the:main} with a constant $c=1$.

\bibliographystyle{IEEEtran}
\bibliography{proposal}

\end{document}